\documentclass[letter]{article} 
\usepackage{amsmath,amssymb,amsthm,graphicx}
\usepackage[colorlinks=true,citecolor=black,linkcolor=black,urlcolor=blue]{hyperref}

\renewcommand{\leq}{\leqslant}

\renewcommand{\le}{\leqslant}
\renewcommand{\ge}{\geqslant}

\newcommand{\Oh}{\mathcal{O}}
\newcommand{\E}{\mathsf{E}}
\newcommand{\G}{\mathcal{G}}
\newcommand{\e}{\mathrm{e}}
\newcommand{\I}{\mathbf{I}}
\newcommand{\U}{\mathbf{U}}

\newcommand{\NN}{\mathbf{N}}
\newcommand{\Or}{\overline{\rho}}

\renewcommand{\Pr}[1]{\ensuremath{\operatorname{\mathbf{Pr}}\left[#1\right]}}

\newcommand{\Ex}[1]{\ensuremath{\operatorname{\mathbf{E}}\left[#1\right]}}
\newcommand{\Var}[1]{\ensuremath{\operatorname{\mathbf{Var}}\left[#1\right]}}

\newcommand{\vol}[1]{\ensuremath{\operatorname{\mathtt{vol}}(#1)}}

\usepackage{setspace,color}   

\newtheorem*{theorem*}{Theorem}
\newtheorem{theorem}{Theorem}  

\newtheorem{lem}[theorem]{Lemma}
\newtheorem{definition}{Definition}
\newtheorem{obs}[theorem]{Observation}
\newtheorem{clm}[theorem]{Claim}
\newtheorem{cor}[theorem]{Corollary}

\newtheorem{rem}[theorem]{Remark}

\usepackage{bbm,lipsum}

\numberwithin{theorem}{section}

\usepackage{microtype}

\title{Tight Analysis of  Asynchronous Rumor Spreading in Dynamic Networks \thanks{ The  authors are supported by the Australian Research Council Discovery Project DP170102794.}}
\author{
	Ali Pourmiri\thanks{Macquarie University, Sydney, Australia, \texttt{ali.pourmiri@mq.edu.au}}
	\and Bernard Mans\thanks{Macquarie University, Sydney, Australia, \texttt{bernard.mans@mq.edu.au}} }

\date{April 24, 2020}

\usepackage{color}   
\definecolor{purple}{rgb}{0.0, .2, .4}

\begin{document}
	\maketitle
	
	\begin{abstract}
		The asynchronous rumor spreading algorithm  propagates a piece of information, the so-called rumor, in a network.
		Starting with a single informed node,  each node is associated with an  exponential time clock with rate $1$ and
		calls a random neighbor in order to possibly exchange the rumor.  A well-studied parameter associated with the algorithm is the \emph{spread time}, which  is the first time when all nodes of a network are informed with high probability\footnote{Event $\mathcal{E}_n$ holds with high probability (w.h.p.) if  $\Pr{\mathcal{E}_n}=1-n^{-c}$, for any given  constant $c>1$.}.
		We consider  spread time of the algorithm in 
		any dynamic evolving network,  $\G=\{G^{(t)}\}_{t=0}^{\infty}$, which is a sequence of arbitrary graphs with the same set of nodes exposed at discrete time step $t=0,1\ldots$.   We observe that besides the expansion profile of a dynamic network, the degree distribution of nodes over time effect the spread time. We establish upper bounds for the spread time in terms of graph conductance and \emph{diligence}.  For a given connected simple graph $G=(V,E)$, the diligence of cut set $E(S, \overline{S})$ is defined as 
		$\rho(S)=\min_{\{u,v\}\in E(S,\overline{S})}\max\{\bar{d}/d_u, \bar{d}/d_v\}$ where $d_u$ is the degree of $u$ and $\bar{d}$ is the average degree of nodes in the one side of the cut  with smaller volume (i.e., $\vol{S}=\sum_{u\in S}d_u$). The diligence of $G$  is also defined as $\rho(G)=\min_{ \emptyset\neq S\subset V}\rho(S)$. For some positive number $\rho$, $G$ is called $\rho$-diligent if $\rho(G)\ge \rho$.
		
		 We show that the spread time of  the algorithm in $\G$ is bounded by $T$,  where $T$ is the first time that
		 $\sum_{t=0}^T\Phi(G^{(t)})\cdot\rho(G^{(t)})$
		   exceeds $C\log n$, where $\Phi(G^{(t)})$ denotes the conductance of $G^{(t)}$ and  $C$ is a specified constant. Moreover, 
		for every $1/\sqrt{n}\le \rho\le 1$, we present a sequence of $\rho$-diligent graphs $G^{(0)},G^{(1)},\ldots$ where  the upper bound matches the spread time up to  $o(\log^2 n)$ factor. 
		
		We also define the \emph{absolute diligence} as $\overline{\rho}(G)=\min_{\{u,v\}\in E}\max\{1/d_u,1/d_v\}$. 
	 We present  upper bound $T$ for the spread time in terms of absolute diligence, which is the first time when $\sum_{t=0}^T\lceil\Phi(G^{(t)})\rceil\cdot \Or(G^{(t)})\ge 2n$. Similarly, we construct dynamic networks where the upper bound is tight up to a constant.   
	  Since for every nonempty graph $G$, $\overline{\rho}(G)\ge {1/(n-1)}$, we conclude that the spread time is bounded by $\Oh(n^2)$ in connected dynamic networks.

		Additionally, we show  that, unlike  static networks, there are striking dichotomies between the spread time of the standard asynchronous  and synchronous algorithms   in  dynamic networks. Hence, one cannot generally estimate the spread time of  synchronous in terms of asynchronous algorithms or vice versa.

	\end{abstract}

\section{Introduction}\label{intro}
Randomized rumor spreading algorithms spread a piece of information, the so-called rumor, in a given network. Initially, an arbitrary single node becomes aware of a rumor, then the algorithm proceeds in synchronized rounds. In each round,  nodes contact a random neighbor and they exchange the rumor if at least one of them knows it, which is known as the push-pull algorithm. 
Demers et. al \cite{DGH+87} first introduced the randomized rumor spreading algorithms to consistently distribute an update in a network of databases. Besides the algorithm simplicity, it  is scalable  in terms of the network size and robust to the node/link failure (e.g., see \cite{FPRU90}). Later on, the rumor spreading  algorithms have been applied in a wide range of settings such as distributed averaging \cite{Boyd2006}, resource discovery \cite{Harchol-Balter1999}, and etc..
The {\it spread time} is a well-studied parameter associated with the rumor spreading algorithms which is the first time when all nodes have been informed with high probability. The spread time of the push-pull algorithm has been studied   on various network topologies, \cite{DFF11,BEF08,FP10}. 
In  \cite{CGLP18}, it has been shown that the spread time of the push-pull in any static $n$-node network  is at most $\Oh(\log n/\phi)$, where $\phi$ denotes the conductance of the network.  In many   distributed networks such as peer-to-peer, social and ad hoc networks, the nodes may not act in a synchronized manner and hence seeking  a more realistic and applicable  mechanism, Boyd et al. \cite{Boyd2006} proposed the asynchronous randomized rumor spreading algorithm. Here, each node has its own  exponential time clock of rate $1$ and  contacts a random neighbor according to arrival times of its Poisson process with  rate $1$. In contrast to  both synchronous and asynchronous  rumor spreading algorithms that have been investigated in various static network typologies \cite{ACMW15,GNW16}, we have known much less regarding randomized rumor spreading algorithms in dynamic networks.
 Here we consider popular \emph{dynamic  evolving network},
  which is a sequence of graphs  $\G=\{G^{(t)}\}_{t=0}^\infty$ arriving  at a sequence of discrete times, $t=0,1\ldots$. They all have the same set of nodes of size $n$, but they may have a different set of edges.
 The model  has gained popularity  (e.g., \cite{AKL18,SL19} ) as it captures various features of  real world networks such as  mobile wireless communication networks, where the set of devices is  unchanged but their relative proximity  changes over time.

 \subsection{ Our Contribution }  
   We consider the asynchronous rumor spreading algorithm in dynamic network $\G=\{G^{(t)}\}_{t=0}^\infty$, where each node  has been associated with an exponential  time clock of rate  $1$.
 Then, as soon as the exponential time clock of a node ticks,  the node picks a random neighbour and they exchange the rumor if at least one of them knows it. Since the  network is dynamic, the underlying communication network may change at discrete time steps.
  Unlike the static networks where the algorithm spreads a rumor  in any connected network after  at most $\Oh(n\log n)$ time with high probability \cite{ACMW15}, we will see that the algorithm spreads the rumor in any connected dynamic networks after at most  $\Oh(n^2)$ time and there are connected dynamic networks where the algorithm spreads the rumor in $\Theta(n^2)$ time, with high probability.
  { It is well-known that any two adjacent nodes communicate within an exponential time distribution with rate $1/d_u(\tau)+1/d_v(\tau)$, where $d_u(\tau)$ is the degree of $u$ at time $\tau$.   Using the order statistics of exponential time distribution the first node after time $\tau\in [0,\infty)$ gets informed within an exponential time distribution with rate $\lambda(\tau)$, which is 
  	\begin{align}\label{intuit}
  	\lambda(\tau)=\sum_{\{u,v\}\in E(\I_\tau, \U_\tau)}\left\{\frac{1}{d_u(\tau)}+\frac{1}{d_{v}(\tau)}\right\},
  	\end{align}  where $E(\I_\tau, \U_\tau)$ is the set of edges crossing  set of informed nodes until time $\tau$ (i.e., $\I_\tau$) and non-informed nodes $\U_\tau$.
As the underlying communication network changes over time,  besides the expansion profile of set of informed nodes, the degree distribution nodes also   directly effect on speed of  the rumor spreading in dynamic network. Therefore,   
 }
   in order to analyze   the spread time of the asynchronous algorithm in dynamic evolving networks, we introduce the notion of \emph{diligence} and  \emph{absolute diligence} of a given connected network $G=(V, E)$. For every $\emptyset\neq S\subset V(G)$, let 
   $E(S,\overline{S})$ be the set of edges crossing $S$ and its complement, $\overline{S}$. Also, define  
    $\vol{S}=\sum_{u\in S}d_u$, where  $d_u$ is the degree of $u$ in $G$, and $\vol{G}=\sum_{u\in V}d_u$. 
    For every $S\subset V$ with $0<\vol{S}\le \vol{G}/2$,  the diligence of cut $E(S,\overline{S})$ is defined as
     \[
     \rho(S)={\min_{\{u,v\}\in E(S,\overline{S})}\max\{\bar{d}(S)/d_u, \bar{d}(S)/d_v\}},
    \]
    where $\bar{d}(S)=(\sum_{u\in S}{d_u})/|S|$.
 The  diligence  of $G$ is defined as
\[
 {\rho}(G)=\min_{\substack{S\subset V\\0<\vol{S}\le \vol{G}/2}}\rho(S),
 \]
 where we always have $1/(n-1)\le \rho(G)\le 1$, provided $G$ is connected. We define $\rho(G)=0$, if $G$ is not connected. 
 For the sake of brevity for every $G^{(t)}\in \G$, we denote $\rho(G^{(t)})$ by $\rho(t)$ and $G^{(t)}$ is called $\rho$-diligent for every $\rho(t)\ge \rho$.
 
 
  Suppose that  $G=(V,E)$ is a nonempty graph, \emph{absolute diligence} of $G$ is defined as  
  \[
  \overline{\rho}(G)=\min_{\{u,v\}\in E}\max\{1/d_u,1/d_v\}.
  \] Similarly,  for every  $G^{(t)}\in \G$ we use  $\overline\rho(t)$ to denote $\overline\rho(G^{(t)})$ and $G^{(t)}$ is called \emph{absolute $\rho$-diligent, for every $\Or(t)\ge \rho$}.
 
    One may observe that if $\G=\{G^{(t)}\}_{t=0}^\infty$ is a sequence of stars (i.e., complete bipartite graph $K_{1, n-1}$), then  every $G^{(t)}\in \G$ is  $1$-diligent and absolutely $1$-diligent, as well. If the dynamic network $\G=\{G^{(t)}\}_{t=0}^\infty$ is regular, which means
  for every  $t=0,1,\ldots$, $G^{(t)}$ is $\Delta_t$-regular. 
  Then every $G^{(t)}\in \G$ is  $1$-diligent. Recall that  for every network $G$, the\emph{ conductance of  $G$}, denoted by $\Phi(G)$, is defined as 
       \begin{align}\label{conduct}\Phi(G)=\min_{\emptyset \neq S\subset V(G)}\frac{|E(S,\overline{S})|}{\min\{\vol{S},\vol{\overline{S}}\}}.
       \end{align}
{       
By definition of the graph conductance and diligence, one can easily obtain a lower bound for $\lambda(\tau)$ (defined in Equation (\ref{intuit})) in terms of those graph parameters as follows. Let $S$ denote the smaller side of cut $E(\I_\tau, \U_\tau)$ in terms of its volume, which is either $\I_\tau$ or $\U_\tau$ and $\vol{S}=\min\{\vol{\I_\tau}, \vol{\U_\tau}\}$.
\begin{align}\label{intu2}
\lambda(\tau)&\ge \sum_{\{u,v\}\in E(\I_\tau, \U_\tau)}\max\{1/d_u(\tau),1/d_{v}(\tau)\}\nonumber\\
&\ge  \rho(S) \frac{|S|}{\vol{S}}|E(\I_\tau, \U_\tau)|\nonumber\\
&\ge \rho(\tau)\cdot {|S|}\cdot  \Phi(G^{(\tau)}) \min\{\vol{\I_\tau}, \vol{\U_\tau}\}/{\vol{S}}\nonumber\\
& \ge \Phi(G^{(\tau)})\cdot \rho(\tau)\cdot
\min\{|\I_\tau|, |\U_\tau|\},
\end{align}
where $G^{(\tau)}$ refers to the underlying graph at time $\tau$ and $|S|\ge \min\{|\I_\tau|, |\U_\tau|\}$.        
 Roughly speaking, the lower bound tells us in each unit time interval the expected number of informed nodes at least increases  by a  multiplicative  factor  $\Phi(G^{(\tau)})\cdot \rho(\tau)$ of $\min \{|\I_\tau|, |\U_\tau|\} $. Having applied the above inequality  and   the \emph{theory of non-homogeneous Poisson processes}  we  present the following theorem regarding the   spread time of the algorithm in  dynamic evolving networks. }

     \begin{theorem}\label{thm:acc}
     	Suppose that  $\G=\{G^{(t)}\}_{t=0}^\infty$ be an $n$-node dynamic evolving network. Also, assume that  a node of $G^{(0)}$ is aware of a rumor. Let $c>1$ be an arbitrary constant and  define 
     	 \[
     	 T(\G,c)=\min\left\{t:\sum_{p=0}^t\Phi(G^{(p)})\cdot\rho(p)\ge {C\log n} \right\},
     	 \]
     	 where $c_0=1/2-1/\e$, $C=(10c+20)/c_0$, and $\rho(p)=\rho(G^{(p)})$.
     	 Then, with probability $1-n^{-c}$, the algorithm propagates the rumor through $\G$ within at most  $T(\G,c)$ time.    
       \end{theorem}
We also construct a family of dynamic networks for which the spread of the algorithm almost matches the presented upper bound.
 \begin{theorem}\label{thm:rho}
 For every given $\frac{1}{\sqrt{n}}\le \rho\le 1$,	there exists $n$-node dynamic evolving network  ${\G}(n,\rho)=\{G^{(t)}\}_{t=0}^\infty$ such that (1) $G^{(t)}\in \G$ is  $\Theta(\rho)$-diligent  and (2)
 	 for every  $G^{(t)}\in {\G}(n,\rho)$,  $\Phi(G^{(t)})=\Theta(1/(k+n\rho^2))$, where  $k=\Theta(\log n/\log\log n)$. Suppose that a node of $G^{(0)}$ is aware of a rumor and the  algorithm starts spreading the rumor through ${\G}(n,\rho)$. 
 	 Then, with high probability, the rumor spreads in $\Omega(n/(\rho\cdot k))$ time. Moreover, this shows that the upper bound for the spread time obtained by Theorem \ref{thm:acc}  is at most $o(\log^2 n )$ factor above the spread time of  the algorithm in ${\G}(n,\rho)$.   
 	 \end{theorem}
We have to remark that by the construction of ${\G}(n,1/\sqrt{n})$, for every $G^{(t)}\in{\G}(n,1/\sqrt{n})$, we have $\Phi(G^{(t)})=\Theta(\log\log n/\log n)$. However, with high probability,  the algorithm spreads the rumor in at least $\Omega(\sqrt{n}/\log n)$ time, which shows the direct impact of the  the graph diligence on the spread time.  
 Next, we establish an upper bound for the spread time of the algorithm in terms of  absolute diligence of networks.
 \begin{theorem}\label{thm:n2}
 	Suppose that     $\G=\{G^{(t)}\}_{t=0}^\infty$ is 
 	an $n$-node dynamic network. Also, assume that  a node of $G^{(0)}$ is aware of a rumor. Define, 
 	\[
 	T_{abs}(\G)=\min\left\{t:\sum_{p=0}^t\lceil\Phi(G^{(p)})\rceil\cdot \Or(p)\ge {2n} \right\},
 	\]
 	where $\lceil\Phi(G^{(p)})\rceil=1$ if $G^{(p)}$ is connected, and zero otherwise.
 	Then, with high probability, the algorithm propagates the rumor through $\G$ within at most $T_{abs}(\G)$ time.
 \end{theorem}
\begin{rem}
If every $G^{(t)}\in \G$ is  connected, then it is  absolutely $\frac{1}{n-1}$-diligent and by the above theorem, the algorithm spreads  rumor in $\Oh(n^2)$ time.  
\end{rem}
 In the following theorem, we construct a family of dynamic networks for which the spread time of the algorithm matches the upper bound, obtained in the above theorem, up to a constant factor. 
  \begin{theorem}\label{thm:n2tight}
  	For every $10/n\le \rho\le 1$,	there exists  $n$-node dynamic network  $\overline{\G}(\rho,n)=\{G^{(t)}\}_{t=0}^\infty$ such that (1)   every  $G^{(t)}\in \G$ is absolutely $\Theta(\rho)$-diligent, (2)   
  	for every $G^{(t)}\in\overline{\G}(n,\rho)$, $\Phi(G^{(t)})>0$ (i.e. $G^{(t)}$ is connected). Suppose that  a node of $G^{(0)}$ is aware of a rumor. Then, the   algorithm propagates the rumor through $\overline{\G}(n, \rho)$ in at least $\Omega(n/\rho)$ time, with  probability $1-O(1/n)$.  
  \end{theorem}

Combining Theorems \ref{thm:acc} and \ref{thm:n2} gives the following Corollary.
\begin{cor}
Suppose $\G=\{G^{(t)}\}_{t=0}^\infty$	is an $n$-node dynamic network   Also, let a node of  $G^{(0)}$ knows a rumor. Then for some arbitrary constant $c>1$, the   spread time of the algorithm in $\G$ is bounded by $
\min\{T(\G,c), T_{abs}(\G)\}$.
  \end{cor}

\begin{figure}[t]
	\vspace{-.5cm}
	\includegraphics[width=0.6\textwidth]{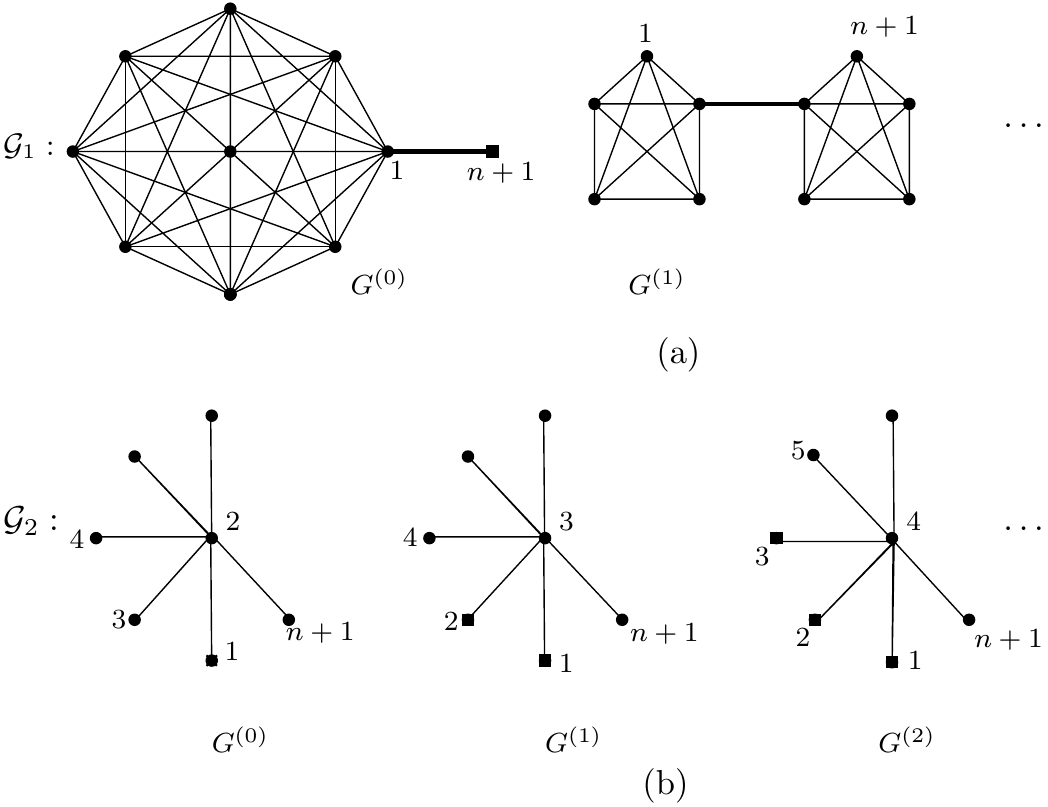}
	\vspace{0cm}
	\caption{\small{\small  $G^{0}\in \G_1$ is an $n$-node clique with a pendent edge $\{1,n+1\}$ and node $n+1$ knows a rumor. 
			$G^{(1)}\in \G_1$ consists of two equally-sized clique joint by an edge, where the left and right clique have nodes $1$ and $n+1$, respectively. Moreover,  for every $t\in \mathbb{N}$,  $G^{(t)}=G^{(1)}$.  $G^{(0)}\in \G_2$ is a star over $n+1$ nodes and leaf-node $1$ knows a rumor. Each $G^{(t)}\in \G_2$, $t\in \NN$, is a star whose center is set be  a non-informed node. If there is no any non-informed node, the center is chosen arbitrarily.
			The square and round  nodes indicate the informed and uninformed nodes, respectively.
	}}\label{Fig1}
\end{figure}
Recently, Giakkoupis et al. \cite{GNW16} applied coupling techniques and  established an  interesting relation between synchronous and asynchronous rumor spreading algorithms.
Let $G$ be a given static network with $n$ nodes and assume that $T_s(G)$ and $T_a(G)$ are the spread time of  synchronous and the standard asynchronous rumor spreading algorithms  on $G$, respectively. They   showed  that
$T_a(G)=\Oh(T_s(G)+\log n)$. Moreover, they  derived an upper bound for  $\frac{T_s(G)}{T_a(G)}$, which is  $n^{1/2}(\log n)^{\Oh(1)}$. However, we will show  such a relation does not exist for rumor spreading algorithms in dynamic networks. Here, we   demonstrate  striking dichotomies between two synchronous and   asynchronous  algorithms in two dynamic networks, namely $\G_1$ and $\G_2$, presented in  Figure \ref{Fig1} (a) and (b).

\begin{theorem}\label{thm:asynch-star}
	Suppose that $T_a(.)$ and $T_s(.)$ denote the spread time of the synchronous and  asynchronous algorithms. (i) $T_a(\G_1)=\Omega(n)$ and $T_s(\G_1)=\Theta(\log n)$, (ii) $T_a(\G_2)=\Theta(\log n)$ and $T_s(\G_2)= n$, and (iii) the  algorithm propagates the rumor in the dynamic star (i.e., $\G_2$) in $2k$ time with probability at least $1-\e^{-k/2-o(1)}-\e^{-k-o(1)}$.  
\end{theorem}
  The dichotomies  shows that  we cannot generally estimate the spread time of the  asynchronous algorithm in terms of the synchronous one or vice versa in dynamic networks.

\subsection{Related works}

  In a closely related work, Giakkoupis et al. \cite{GSS14} studied the spread time of  the synchronous push-pull algorithm  in  dynamic evolving  network $\G=\{G^{(t)}\}_{t=0}^\infty$ with a vertex set of size $n$. 
   Here, the authors show that, with high probability,  the spread time of the algorithm is bounded by  $\min\{t: \sum_{p=0}^t\Phi(G^{(p)})=\Omega(M(\G)\log n)\}$, where $M(\G)=\max_{u}\Delta_{u}/\delta_u$, $\Delta_u$ and $\delta_u$ are the upper and lower  bounds for degree of node $u$ over all time steps $t$  and the maximum is  taken over all nodes.
  {Additionally,  they presented some dynamic networks for which the upper bound is tight. However, if we consider a dynamic network $\G=\{G^{(t)}\}$, where $G^{(t)}$ is $d(t)$-regular graph with $d(t)\in \{3, n-1\}$, and every other graph is a complete graph. Then, the upper bound will be of order $\Oh(n\log n)$, which is $\Oh(n)$ factor above the real spread time. 
  		They obtained the upper bound based on the   growth rate analysis of $\vol{\I_t}$ over time and hence it causes to have $M(\G)$ factor in the upper bound. In our work we introduce the notion of graph diligence in order to directly keep track of growth rate of $\I_t$ and establish a relatively more accurate  bound.  }

   
      Apart from this work, the information dissemination has been studied on evolving network models, where the evolution in the model is governed by a random process.  Clementi et al. \cite{CCDFIPPS13} considered the push algorithm on  edge-Markovian evolving graph model. In this model, given positive numbers $p, q, \in (0, 1)$ and an initial graph,   a non-edge appears with probability $p$ and an edge dies with probability $q$. They proved that when $p=\Omega(1/n)$ and $q$ is constant, the push algorithm propagates the rumor in $\Oh(\log n)$ rounds, w.h.p. 
  {\it Flooding} is a simple variant of the rumor spreading algorithms in which in every time step an informed node sends the rumor to all its neighbors. In    \cite{CST12,CMPS11,BCF09}, 
 they have studied the spread time of flooding in an ergodic Markovian dynamic graph process, i.e. when the network topology at time $t$ only depends on the topology at time $t-1$.
 In \cite{PPPU11,LLMSW12},   authors considered a dynamic evolving network, whose nodes are mobile agents performing random walks on a $2$-dimensional grid and a piece of  information is transmitted from one agent to another when they are sufficiently close to each other (and at least one of them has the information). They obtained an upper bound for the spread time of the process  in terms of number of agents and the size of the grid.

 \paragraph{Outline}{In Section \ref{sec:not}, we provide some notations and preliminaries. In Section \ref{sec:analysis} we state   our upper bound for the spread time of  algorithm in terms of conductance and diligence of networks (i.e., Theorem \ref{thm:acc}).
 In Section \ref{sec:tight} we prove Theorem \ref{thm:rho}. In Section \ref{sec:abs}, we prove Theorems \ref{thm:n2} and \ref{thm:n2tight}. In Section \ref{asyvs}, we prove Theorem \ref{thm:asynch-star}.

\section{Notations and Preliminaries }\label{sec:not}
In this section we first define notations and some useful preliminaries. 
 We use $\G=\{G^{(t)}=(V,E_t)\}_{t=0}^\infty$ to denote a sequence of simple graphs, where all $G^{(t)}$'s have the same set of nodes of size $n$, denoted by $V$, but arbitrary set of edges. Let $\mathbf{N}$ to be the set of non-negative integers (i.e., $\mathbb{N}\cup\{0\}$).
 For every $\emptyset\neq S\subseteq V$, $E_t(S, \overline{S})$ denotes the set of edges crossing set $S$ and its complement, $\overline{S}$, in $G^{(t)}$. Also, for each $t\in \NN$ and node $u$, $d_u(t)$ denotes the degree of $u$ in $G^{(t)}$.
 Throughout this paper, $n$ denotes the number of nodes in the dynamic network. 
 We say an event, say $\mathcal{E}_n$ holds with high probability, if $\Pr{\mathcal{E}_n}\ge 1-n^{-c}$, for every given constant $c>1$. 
 In order to analyze the algorithm we  keep track of the size of informed and non-informed nodes in at continuous time interval $[0, \infty)$.  
 We usually  use $t$ and $\tau$ to denote the discrete time steps and continuous time, respectively. For every $\tau\in[0, \infty)$, $\lfloor \tau \rfloor$ ($\lceil \tau \rceil$)  denote the biggest (smallest) before (after) $\tau$.     
 Let
$\I_\tau$ and  $\U_\tau$ denote the set of informed and uninformed nodes until time $\tau$, respectively, and their sizes are indicated by $I_\tau$ and $U_\tau$.  We have to remark that   any graph property including node degree, cut, and conductance at continues time $\tau$ refers to $G^{(\tau)}=G^{(\lfloor \tau\rfloor)}$. 

 Let us  now formally define the asynchronous algorithm in dynamic evolving networks. 
   \begin{definition}[The Asynchronous Algorithm in Dynamic Evolving Network $\G=\{G^{(t)}\}_{t=0}^{\infty}$]\label{def:AA}
  	Suppose at time $t=0$, some node of a given network, say $G^{(0)}$, is aware of a rumor. Also,
  	assume that each node  has been associated with an exponential time clock of rate $1$,  
  	Then, each node   contacts a random neighbor according to the arrival times of a Poisson process with rate $1$.  When they contact each other, they may learn the rumor, if at least one of them knows it (i.e., node sends (pushes) the rumor to the neighbor if he knows it, or the node asks (pulls) the rumor from the neighbor and he may learn it). 
  	Notice  that, at succeeding discrete time steps $t\in \NN$, the underlying network may change and an arbitrary graph is exposed.  We define the spread time as the first time when all nodes become informed with high probability.
   \end{definition}
	In order to analyze the algorithm in dynamic networks, we apply the theory of the non-homogeneous Poisson process.
  \begin{definition}[Non-homogeneous Poisson process]\label{nhpp}
  	Suppose that there is an exponential time clock whose rate, say $\lambda(\tau)$,  is a non-negative function of time $\tau\in [0, \infty)$ (i.e.,  $\lambda(\tau): [0, \infty)\rightarrow [0, \infty)$). 
  	Also let $N(\tau)$ to denote the number of ticks made by the clock in $[0, \tau]$. Then, $\{N(\tau): \tau\ge 0\}$ is called a non-homogeneous Poisson process with rate $\lambda(\tau)$.  
  \end{definition}

\begin{theorem}\cite[Chapter 2]{NPP03}\label{thm:npp}
	Suppose that  $\{N(\tau): \tau\ge 0\}$ is  a non-homogeneous Poisson process with rate $\lambda(\tau)$. Also assume that  $\lambda(\tau):[0,\infty)\rightarrow[0, \infty)$ is an  integrable function. Then, for every $0\le a\le b$, $N(b)-N(a)$ has a Poisson distribution with rate 
	\[
	\Lambda=\int_{a}^b\lambda(\tau)d\tau.
	\]
\end{theorem} 
   For more information about non-homogeneous Poisson processes we refer the interested reader to \cite{NPP03}. 
    We  state the following useful observation.
   \begin{lem}\label{ineq:pois}
     	Suppose that $X$ is a Poisson random variable with rate $r$. Then,
     	$
     	\Pr{X\le r/2}\le \e^{r(1/\e+1/2-1)}$.
     \end{lem}
  \begin{proof}
 	\begin{align*}
 		&\Pr{X\le r/2} 
 		=\e^{-r}\sum_{j=0}^{r/2}\frac{2^j(r/2)^j}{j!}
 		\le \e^{-r}\sum_{j=0}^{r/2}\frac{2^{r/2}(r/2)^j}{j!}\\
 		&=\e^{-r}\e^{\frac{r\log2}{2}}\sum_{j=0}^{r/2}\frac{(r/2)^j}{j!}
 		\le\e^{r(\frac{\log(2)}{2}-1)}\e^{\frac{r}{2}}
 		=\e^{r(\frac{\log2}{2}-1+\frac{1}{2})}
 		\le \e^{r(1/\e+1/2-1)},
 	\end{align*}
 	which follows from $(\log 2)/2<1/\e$.
 \end{proof}

     \section{  The Spread Time in  Dynamic Networks }\label{sec:analysis}
In this section we prove Theorem \ref{thm:acc}.
Let us first recall that for every $G^{(t)}\in \G$ and $S\subset V$, with $0<\vol{S}\le \vol{G^{(t)}}/2$, 
$
\rho(S)=\min_{\{u,v\}\in E(S, \overline{S})}\max\{\bar{d}(S)/d_u(t),\bar{d}(S)/d_v(t)\}$, where $\bar{d}(S)=\vol{S}/|S|=(\sum_{u\in S}d_u)/|S|$. Also,

\begin{align}\label{dilig}
{\rho}(G^{(t)})=\min_{\substack{S\subset V\\0<\vol{S}\le \vol{G}/2}}\rho(S).
\end{align}
Also we  assumed that $\rho(G)=0$, if $G$ is not connected.
For a dynamic network $\G=\{G^{(t)}\}_{t=0}^\infty$, we use  $\rho(t)$ to denote $\rho(G^{(t)})$ and  $G^{(t)}$ is called a $\rho$-diligent graph, if $\rho(t)\ge \rho$.
 The following lemma plays a key role in the proof of Theorem \ref{thm:acc}.
\begin{lem}\label{lem:key}
  Suppose that $\G=\{G^{(t)}\}_{t=0}^\infty$ be a dynamic network  and initially a rumor is injected to some node of $G^{(0)}$. Suppose that we are at some   arbitrary time $\tau\in[0, \infty)$ and there are $I_\tau$ informed and $U_\tau$  uninformed  nodes.
   Define  
 \[
 \Delta(\alpha)=\min\left\{q: \sum_{p=0}^q\Phi(G^{( \lceil\tau\rceil+p)})\cdot{\rho(\lceil\tau\rceil+p)}~\ge~ {2\alpha}\right\}.\]
  Let $\tau'$ denotes the earliest time when the number of informed nodes increases by $\min\{I_\tau,U_\tau\}/2$.
 Then, 
\[
\Pr{\tau'-\tau\ge \Delta(\alpha)+2}\le \e^{-c_0\alpha\min\{I_\tau, U_\tau\}},
\]
 where  $c_0=1-1/2-1/\e$.
\end{lem}
 \begin{proof}
 For $\gamma\in[\tau, \tau')$, let $E({\I_\gamma},{\U_\gamma})$ denote the set of edges crossing the set of informed nodes $\I_\gamma$ and non-informed nodes $\U_\gamma$, in $G^{(\gamma)}=G^{(\lfloor \gamma\rfloor)}$. For every  $\{u, v\}\in E({\I_\gamma},{\U_\gamma})$, informed node $u$ pushes the rumor to $v$ with rate $1/d_u(\gamma)$ and  non-informed $v$ pulls it from $u$ with Poisson rate $1/d_v(\gamma)$.
{Therefore, $u$ and $v$ contact each other with rate $
 	{1}/{d_u(\gamma)}+{1}/{d_v(\gamma)}$,
 	which follows from the fact that minimum of  independent exponential random variables is an exponential random variable whose rate is the sum of those rates.} 	 Furthermore, applying the order statistics of independent  exponentially  distributed  random variables implies that the first non-informed node after time $\gamma$ gets informed within an exponential distribution with rate  
 	 \begin{align*}
 	 \lambda(\gamma)=\sum_{\{u,v\}\in E({\I_{\gamma}},{\U_{\gamma}})} \left\{\frac{1}{d_u(\gamma)}+\frac{1}{d_v(\gamma)}\right\}
  \end{align*}
Let  $S=\I_\gamma$ if $\vol{I_\gamma}\le \vol{U_\gamma}$ and $S=\U_\gamma$, otherwise. Also note that  $\rho(S)$ is the diligence of cut $E(S, \overline{S})$ in $G^{(\gamma)}$.  Then by the definition of $\rho(S)$, we get that  

 \begin{align}\label{lambdat}
	\lambda(\gamma)\ge \rho(S)\frac{|S|}{\vol{S}}|E({\I_{\gamma}},{\U_{\gamma}})|\ge \rho(\gamma)\frac{|S|}{\vol{S}}|E({\I_{\gamma}},{\U_{\gamma}})|.
\end{align}
By the definition of graph conductance (e.g., see (\ref{conduct})) we get
\begin{align*}
	\lambda(\gamma)&\ge  \rho(\gamma)\cdot |S| \Phi(G^{(\gamma)})\min\{\vol{\I_\gamma},\vol{\U_\gamma}\}/\vol{S}\nonumber\\
	&\ge\Phi(G^{(\gamma)})\rho(\gamma) |S| \ge \Phi(G^{(\gamma)})\rho(\gamma)  \min\{I_\gamma, U_\gamma\}, 
\end{align*}
where $G^{(\gamma)}=G^{(\lfloor\gamma\rfloor)}$ and the  last inequity follows from  $|S|\ge \min\{I\gamma, U_\gamma\}$.
Let $m(\tau)=\min\{I_\tau, U_\tau\}$. By the definition of $\tau'$ in the lemma statement, for every $\gamma\in [\tau, \tau']$, the non-informed nodes decreases by a factor $1/2$ and informed may increase as well. So  
\begin{align}\label{last}
	\lambda(\gamma)\ge \Phi(G^{(\gamma)})\cdot\rho(\gamma)\cdot \min \{I_\gamma, U_\gamma\}\ge 
	\Phi(G^{(\gamma)})\cdot\rho(\gamma)\cdot m(\tau)/2
	.
\end{align}
 	 Let us define a non-homogeneous Poisson process $\{N(\gamma), \gamma\ge \tau\}$ with rate $\lambda(\gamma)$, where $N(\gamma)$ counts the number of occurrences of the non-homogeneous Poisson process in time interval $[\tau, \gamma)$. Equivalently, $N(\gamma)$  counts the number of newly informed nodes in time interval $[\tau, \gamma)$.

 Note that 
 if $\tau'-\tau\le \Delta(\alpha)+2$, then  $\Pr{\tau'-\tau\ge \Delta(\alpha)+2}=0$ and the proof is completed. So let us assume that    $\tau+\Delta(\alpha)+2\le \tau'$. 
 By the definition (i.e., Inequality (\ref{lambdat})), $\lambda(\gamma)$ is integrable. Moreover,  by Theorem \ref{thm:npp}, $N(\tau+\Delta(\alpha)+2)$ has a Poisson distribution with rate
  $\Lambda(\tau+\Delta(\alpha)+2)$, which is at least
   \begin{align*}
\Lambda(\tau+\Delta(\alpha)+2)&=\int_{\tau}^{\tau+\Delta(\alpha)+2}\lambda(\gamma)d\gamma\ge\int_{\lceil\tau\rceil}^{\lceil\tau\rceil+\Delta(\alpha)+1}\lambda(\gamma)d\gamma\\
&=\sum_{p=0}^{\Delta(\alpha)}\int_{\lceil\tau\rceil+p}^{\lceil\tau\rceil+p+1}\lambda(\gamma)d\gamma
\end{align*}
 where the first inequality holds because $[\lceil\tau\rceil, \lceil\tau\rceil+\Delta(\alpha)+1]\subset [\tau, \tau+\Delta(\alpha)+2]$.
 Applying Inequality (\ref{last}) and definition of $\Delta(\alpha)$ in the lemma statement results into
 \begin{align*}
\Lambda(\tau+\Delta(\alpha)+2)&\ge \sum_{p=0}^{\Delta(\alpha)} \left\{\Phi(G^{(\lceil\tau\rceil+p+1)})\rho(\lceil\tau\rceil+p+1) \cdot m(\tau)/2\right\}\\&
\ge 2\alpha m(\tau)/2=\alpha \cdot m(\tau)
 \end{align*}
 
 Therefore, $N(\tau+\Delta(\alpha)+2)$ 
 stochastically dominates a Poisson distribution, say $X$, with rate
$r=\alpha\cdot m(\tau)$. By the definition of $\tau'$ and Lemma \ref{ineq:pois}, we get
\begin{align*}
&\Pr{\tau'-\tau\ge \Delta(\alpha)+2}\le \Pr{ N(\tau+ \Delta(\alpha)+2)
	\le m(\tau)/2}\\&\le\Pr{ X\le m(\tau)/2} \le \Pr{X\le r/2}\le \e^{-c_0\alpha \cdot m(\tau)}, 
\end{align*}
 where $c_0=1-1/2-1/\e$ . 
This shows that, with probability at least  $1-\e^{-c_0\alpha m(\tau)}$, during time interval $[\tau, \tau+\Delta(\alpha)+2)$, the number of informed nodes increases by at least $m(\tau)/2$. 
   	  	\end{proof}
     	\begin{proof}[Proof of Theorem \ref{thm:acc}]
     	The proof is based on iterative applications of Lemma \ref{lem:key}. 
     We analyze the spread time of the algorithm in two consecutive phases.
     	
     	\underline{ First Phase:}
     	This phase starts with $I_0=1$ and ends when the number of informed node  exceeds $n/2$.
     	 Let $\tau_f$ denotes the time when the phase ends.
     	For each integer $0\le i\le \log_{3/2}(n/2)$,  consider sub-interval $[\tau_i, \tau_{i+1}]\subset [0, \tau_f]$ where  during each one, the number of informed nodes increases by additive factor $I_{\tau_i}/2$ and inductively, by the end of the $i$-th interval there are $(3/2)^{i}$ informed nodes.  For each
     $0\le i\le \log_{3/2}(n/2)$, let us define $\alpha_i=\lceil\frac{c\log n}{c_0(3/2)^i}\rceil$. Moreover, 
      for every $\tau\in[0,\tau_f]$, $I_\tau\le n/2\le U_\tau$ and hence $\min\{I_\tau, U_\tau\}=I_\tau$.
      Then, by Lemma \ref{lem:key}, with probability $	1-\e^{-c_0\alpha_i(3/2)^{i}}\ge 1-\e^{-c\log n}$,  $\Delta(\alpha_i)+2$ is an upper bound for $\tau_{i+1}-\tau_{i}$. The union bound over all $0\le i\le \log_{3/2}(n/2)$ implies that with probability  at least $1-(4\log n)\e^{-c\log n}$, we get
     \begin{align}\label{upper:phase1}
     	\tau_f&\le \sum_{i=0}^{\log_{3/2} n/2}\left\{\Delta(\alpha_i)+2\right\}\nonumber\\
     	&\le  5\log n+\min \left\{t:\sum_{p=0}^{t}\Phi(G^{(p)})\cdot\rho(p)\ge {2}\sum_{i=0}^{\log_{3/2} n/2}\alpha_i\right\},
     \end{align}  
     which follows from $\log_{3/2} n/2\le (2\cdot5)\log n$.
     
     \underline{Second Phase:} This phase starts with $ n/2$ informed nodes (or at most $n/2$ non-informed nodes) and ends when there is no any non-informed nodes. Let $\tau_s$ denotes the time when the phase ends. Similar to the first phase, for each integer $1\le j\le \log_2n $, consider sub-interval $[\tau_j,\tau_{j+1}]\subset(\tau_f,\tau_s]$ where during each $[\tau_j,\tau_{j+1}]$ the number of uninformed nodes decrease by $U_{\tau_j}/2$ and inductively,  by the end of the $j$-th interval, the number of non-informed nodes reduces to $n(1/2)^{j+1}$. For each $0\le j\le \log_2 n$, define $\beta_j=\lceil\frac{c\log n}{c_0 n(1/2)^j}\rceil$. In this phase for every $\tau\in (\tau_f,\tau_s]$,  $ U_\tau\le n/2\le I_\tau$ as $U_\tau$ is a non-increasing in terms of $\tau$ and $I_\tau$ is non-decreasing in $\tau$. So we have $\min\{I_\tau, U_\tau\}=U_\tau$.
     Therefore, applying Lemma \ref{lem:key} implies that, with probability $1-\e^{c_0\beta_jn(1/2)^j}\ge 1-\e^{-c\log n}$, we have that   
     $\tau_{j+1}-\tau_j\le \Delta(\beta_j)+2$.
     The union bound over all $1\le j\le \log_{2}(n)$ implies that with probability  at least $1-(\log_2 n)\e^{-c\log n}$, we get
     
     \begin{align}\label{upper:phase2}
     	\tau_s-\tau_f&\le \sum_{j=1}^{\log_{2}n}\left\{\Delta(\beta_j)+2\right\}\nonumber\\
     	&\le  5\log n+\min \left\{t:\sum_{p=0}^{t}\Phi(G^{(p)})\cdot \rho(p)\ge {2}\sum_{j=1}^{\log_{2} n}\beta_j\right\},
     \end{align}
     which follows from $\log_2 n\le 2\log n$.
     The geometric series calculation shows that  
     \[
    {2}\left(\sum_{i=0}^{\log_{3/2}(n/2)}\alpha_i+
     \sum_{j=1}^{\log_{2}(n)}\beta_j\right)\le {(10c+10)\log n}/{c_0}.
     \]
    Using the above upper bound and combining the upper bounds in Inequalities (\ref{upper:phase1}) and (\ref{upper:phase2}) result that with probability $1-n^{-c+o(1)}$, we have 
    
    \begin{align*}
    \tau_s&=\tau_f+(\tau_s-\tau_f)\\
    &\le 10\log n+ \min \left\{t:\sum_{p=0}^{t}\Phi(G^{(p)})\cdot \rho(p)\ge {(10c+10)\log n}/{c_0}\right\}\\&\le 
     \min \left\{t:\sum_{p=0}^{t}\Phi(G^{(p)})\rho(p)\ge {(10c+20)\log n}/{c_0}\right\},
    \end{align*}  
 where $\tau_s$ is an upper bound for the spread time and completes the proof.
     \end{proof}

\section{The Upper Bound Is Almost Tight}\label{sec:tight}
In this section we construct dynamic networks for which the spread time of the algorithm almost matches the upper bound  obtained in Theorem \ref{thm:acc}. We apply them to prove Theorem \ref{thm:rho}. 

Let $V$ denote a set of $n$ nodes and $A\subset V$ be an arbitrary subset, where $n/4\le |A|\le 3n/4$. Also, let $B= V\setminus A$.  Fix arbitrary integer numbers $1\le k=k(n)=\Oh(\log n/\log\log n)$ and $1\le \Delta=\Delta(n)=\Oh(\sqrt{n})$.
We  build graph $H_{k,\Delta}(A,B)$ with vertex set $A\cup B$ in the following two  steps as follows;
\begin{enumerate}
	\item Let $\{S_i, 0\le i\le k\}$ denote a set of disjoint clusters of nodes of $A\cup B$, where for each $ i=0\ldots k$, $|S_i|=\Delta$ and we have  $S_0\subset A$ and $\cup_{i=1}^k S_i\subset B$.
	For each $i=0,\ldots k$, connect each node of $S_i$ to all nodes in $S_{i+1}$. The resulting graph is a string of complete bipartite graphs having $(k+1)\cdot \Delta$ nodes and $k\cdot \Delta^2$ edges.
	\item Let $G_1=(A\setminus S_0, E_1)$ and $G_2=(B\setminus\cup_{i=1}^k S_i, E_2)$ denote two arbitrary $4$-regular expander graphs on vertex sets $A\setminus S_0$ and $B\setminus\cup_{i=1}^kS_i $, respectively. We say graph $G$ is an expander graph if $\Phi(G)=\Theta(1)$. We now  connect each node of $S_0$ to $\Delta$ distinct nodes of $G_1$ such that for every node $u\in A\setminus S_0$, degree of $u$ in $G_1$ increases  by at most an additive constant.   
	Similarly, we connect each node of $S_k$ to $\Delta$ distinct nodes of  $G_2$ so that for every node $u\in B\setminus\cup_{i=1}^kS_i$, degree of $u$ in $G_2$ increases  by at most an additive constant.
	 
\end{enumerate}
We now observe that;
\begin{obs}\label{obs:hn}
	For every graph $H_{k,\Delta}(A, B)$ with $\Delta=\Oh(\sqrt{n})$, the followings hold
	\begin{itemize}
		\item $\Phi(H_{k,\Delta}(A, B))=\Theta\left(\frac{\Delta^2}{k\Delta^2+n}\right)$, 
		\item   $\rho(H_{k,\Delta}(A,B))=\Theta(\frac{1}{\Delta})$
	\end{itemize}
\end{obs}
\begin{proof}
	By the construction of $H_{k,\Delta}(A, B)$, we observe that for each $q=1,...,k$ if we set  $A_q=A\cup(\cup_{i=1}^q S_i)$, then $\vol{A_q}=2(q+1)\Delta^2+\Theta(n)$, because by the construction nodes in  $A\setminus S_0$ have constant degree and nodes in $\cup_{i=0}^qS_i$ have degree $2\Delta$. Therefore,
	\[\Phi(H_{k, \Delta}(A, B))\le \frac{|E(A_q,\overline{S_q})|}{\vol{S_q}}= \frac{\Delta^2}{2(q+1)\cdot \Delta^2 +\Theta(n)}.\]
	Since graphs induced by $A\setminus S_0$ and $B\setminus (\cup_{i=1}^k S_i) $ are expander, a few case analysis implies that  
	\[
	\Phi(H_{k, \Delta}(A, B))=\Theta\left(\frac{\Delta^2}{k\cdot \Delta^2 +n}\right).
	\]
	Consider set $A_1=A\cup S_1$ and cut set $E(A_1,\overline{A_1})$. We know that $S_0\subset A_1$ and $S_1\subset A_1$ so we conclude that  $d(A_1)=(4\Delta^2+\Theta(n))/|A_1|=\Theta(1)$, as $\Delta^2=\Oh(n)$. Then, we have $\rho(A_1)=d(A_1)/{2\Delta}=\Theta(1/\Delta)$. Since the graph has constant average degree and its maximum degree is $\Theta(2\Delta)$. We have $\rho(H_{k, \Delta}(A, B))=\Theta(1/\Delta)$.
\end{proof}

\subsubsection*{\underline{ $\rho$-Diligent Dynamic Network $\G(n,{\rho})$}}
We now
 define evolving dynamic network
$\G(n,{\rho})=\{G^{(t)}=(V, E_t)\}_{t=0}^\infty$, where each $G^{(t)}\in \G$ is $\Theta(\rho)$-diligent; Let $\Delta=\lceil1/\rho\rceil$ and set
 $G^{(0)}=H_{k, \Delta}(A_0, B_0)$, where $A_0$ and $B_0$ are the arbitrary disjoint subsets of $V$ and  $|A_0|=n/4$ and $|B_0|=3n/4$. Let us assume that at time $t=0$, we inject a rumor to a  node of $A_0$. For every $t\in \NN$, define  $B_{t+1}=B_{t}\setminus \I_{t+1}$ and $A_{t+1}=V\setminus B_{t+1} $, where $\I_{t+1}$ is the set of informed nodes up to time $t+1\in \NN$.
If $n/4\le |B_{t+1}|$ and $|B_{t+1}|<|B_t|$, then $G^{(t+1)}=H_{k,  \Delta}(A_{t+1}, B_{t+1})$ and otherwise $G^{(t+1)}=G^{(t)}$.

\begin{lem}\label{lem:string1}
	For some $t\in \NN$, consider graph $G^{(t)}=H_{k, \lceil 1/\rho\rceil}(A_t, B_t) \in \G(n, {\rho})$, where $\lceil 1/\rho\rceil=\Oh(\sqrt{n})$.  Suppose that, at beginning of time $t$,
	at least a node of $A_t$ is aware of a rumor but no node in $B_t$ knows the rumor. The  algorithm starts propagating  the rumor trough $H_{k, \lceil 1/\rho\rceil}(A_t, B_t)$. Then, the probability that at time $t+1$, at least a node of $S_k$ gets informed is at most $(2^k/k!)\cdot\lceil 1/\rho\rceil$.
\end{lem}
\begin{proof}
	In order to inform a node in $S_k$, it is necessary to have at least one informed  node in $S_0\subset A_t$.
	  Let us assume that at time $t$, all nodes contained is $S_0\subset A_t$ are informed. Moreover,  by the construction of $H_{k,\lceil1/\rho\rceil}(A_t,B_t)$, each node in $\cup_{i=0}^k S_i$ has degree $2\lceil1/\rho\rceil=2\Delta$.
For each  $i=0,\ldots,k-1$, let $\{v_i,v_{i+1}\}$ be an arbitrary edge crossing $S_i$ and $S_{i+1}$.  Each node in $\cup_{i=0}^k S_i$ has its own exponential clock of rate $1$ and hence $v_i$ an $v_{i+1}$ communicates via  $\{v_i, v_{i+1}\}$ according to an exponential time clock of rate $\frac{1}{2\Delta}+\frac{1}{2\Delta}=\frac{1}{\Delta}$.
In order to simplify the analysis of the  algorithm in $\cup_{i=0}^k S_i$, we  consider  the asynchronous $2$-push algorithm, defined as follows. First, each node of $\cup_{i=0}^k S_i$ is associated with a clock of rate $2$.  If $u\in \cup_{i=0}^k S_i$ gets informed and its clock rings, then  $u$  sends (pushes) the rumor to a randomly chosen  neighbor.  In the asynchronous $2$-push algorithm, for each $i=0,\ldots,k$,  any edge crossing $S_i$ and $S_{i+1}$ is picked with exponential time clock $2/2\Delta$. Therefore,  the original algorithm and the  $2$-push spreads the rumor in $\cup_{i=0}^k S_i$, equivalently.
Let $\mathcal{E}_{1}$ and $\mathcal{E}_{2}$ denote the events that, during time $[t,t+1]$, the original algorithm and the $2$-push 	delivers the rumor to some node at $S_k$, respectively. Then, we have 
$
\Pr{\mathcal{E}_1}=\Pr{\mathcal{E}_2}.
$
We also define the \emph{forward $2$-push algorithm}.  
In this algorithm each node of $\cup_{i=0}^{k-1} S_i$ is associated with an exponential clock of rate $2$.  If node $u\in S_i$, $i=0,\ldots,k-1$, is informed and its clock rings, then $u$ pushes the rumor to a randomly chosen neighbor in $S_{i+1}$. Note that $u\in S_i$ has $\Delta$ neighbors in $S_{i-1}$ and $\Delta$ neighbors in $S_{i+1}$.  Also, let
 $\mathcal{E}_{3}$ denote the events that, during time $[t,t+1]$, the forward $2$-push algorithm 	delivers the rumor to the first  node in $S_k$, respectively. Then, we have 
\begin{clm} 
	\begin{align}\label{ineq:coup}
		\Pr{\mathcal{E}_{2}}\le \Pr{\mathcal{E}_{3}},
	\end{align}
\end{clm}
\begin{proof}  
Consider the $2$-push algorithm in the graph induced by $\cup_{i=0}^k S_i$, which is a string of complete bipartite graphs and all nodes in $S_0$ are informed. In the both algorithms every node is associated with an exponential time click of rate $2$.  In the $2$-push we  assume that each node  $u\in \cup_{i=0}^k S_i$  has a contact list of neighbors and when $u$ gets informed,  he sends the rumor to them according to the list with specified times, called $C(u)$. Suppose that  for every node
$u\in \cup_{i=0}^k S_i$, $C(u)$ has been  generated.
For every $i=1,\ldots,k-1$, and node $u\in S_i$, $v\in C(u)$ is a forward neighbor of $u$, if $v\in S_{i+1}$ and $v$ is a backward neighbor if $v\in S_{i-1}$ .
Let us now build contact lists of nodes in the forward $2$-push, which is $C_{f}(u)$, $u\in \cup_{i=0}^k S_i$. 
 For $u\in \cup_{i=0}^k S_i$, let us initialize $C_f(u)=\{v\in C(u): \text{$v$ is a forward neihbours}\}$.  For every $u\in \cup_{i=0}^k$, we are now going to update $C_f(u)$ to be  a contact list of $u$ in the forward $2$-push. Assume that  $P$ is an arbitrary  directed  trajectory of the rumor from some node $v_0\in S_0$ to a  node of $v_k\in S_k$ in the $2$-push algorithm.  
Then $P$ is an ordered sequence of arcs, say $a_1a_2\ldots a_q$. Let us label  the set of arcs as follows; \mbox{$L:\{a_1,a_2\ldots,a_q\}\rightarrow \{-,+\}$} so that $L(a_x)=+$ if for some $0\le i\le k-1$, $a_x$ goes from $S_i$ to $S_{i+1}$ and $L(a_x)=-$, otherwise. 
If for every $x\in\{1\ldots q\}$, $L(a_x)=+$, then we do not update the contact list of nodes with respect to $P$ and the initial contact list delivers the rumor. 
Suppose that $P$ contains at least one negative arc.
 Let $y\in\{1\ldots q\}$ be the   smallest index with $L(a_{y})=-$ and let $a_{y}=(u_{y},v_{y})$.
Consider two cases; (1) $v_{y}\in S_0$ in this case we know that all nodes in $S_0$ are informed and the arc dose not inform a new node so we do not update any lists. (2) If for some $i\ge 1$, $v_{y}\in S_{i}$. Since $y$ is the smallest index and $v_{y}\in S_i$,  $a_{y-2}a_{y-1}a_{y}$ is a path of length $3$ from $S_{i-1}$  to $S_{i}$.  Let $a_{y-2}=(u_{y-2}, v_{y-2})$. Since $u_{y-2}$ is in the trajectory, he knows the rumor and we update $C_f(u_{y-2})$ as $C_f(u_{y-2})=C_f(u_{y-2})\cup\{v_y\}$. Let us replace negative arc $a_y$ by $(u_{y-2},u_y)$ in $P$ and go to the next negative arc with minimum index. We continue until no negative arc is left. 
 Considering the updated contact update list, one my observe the list is a random list generated by the forward $2$-push algorithm and   modified $P$ (converting  the negative arcs to positive ones) is  a trajectory in the forward $2$-push that is going to deliver the rumor from $v_0$ to $v_k$ and hence the claim is proved.

 \end{proof}

Suppose that at time $t=0$  the forward $2$-push algorithm starts propagating the rumor in $\cup_{i=0}^k S_i$. By the assumption each node of $S_0$ was informed at time $t=0$.  For every $0\le i\le k$, let  $I(\gamma,i)$ denotes the number of informed nodes contained in $S_{i}$ up to time $\gamma\in[0,1]$.  Then for every $\gamma\in[0,1]$, we have that $I(\gamma,0)=|S_0|=\Delta$. Moreover, for each $i=0\ldots k-1$, $I(\tau,i)$ has $\Delta|I(\tau,i)|$ neighbors in $S_{i+1}$ and each edge is picked by the forward $2$-push with  rate $2/\Delta$. Therefore, by applying Theorem \ref{thm:npp} we get that  
\begin{align*}
\Ex{I(\tau,1)}&=\Ex{ \Ex{I(\tau,1)| I(\gamma,0), \gamma \in [0,1]}}\\
&\le \Ex{\int_{0}^\tau (2/\Delta){I(\gamma,0)}\Delta d\gamma}=\int_{0}^\tau  2\Delta ds= 2\tau\Delta,
\end{align*}
where $\Delta=\lceil 1/\rho\rceil$. Moreover, we have   
\begin{align*}
&\Ex{I(\tau,2)}=\Ex{ \Ex{I(\tau,2)| I(\gamma,1), \gamma\in [0,1]}}\\
&\le \Ex{\int_{0}^t (2/\Delta){I(\gamma,1)}\Delta d\gamma}=\int_{0}^\tau 2\Ex{I(\gamma,1)} d\gamma\\&\le \int_{0}^t 2^2 s\Delta  ds=\frac{2^2\tau^2}{2}\Delta
\end{align*}
Inductively, we get
\[
\Ex{I(\tau,k)}=\Ex{ \Ex{I(\tau,k)| I(\gamma,k-1), \gamma\in [0,1]}}\le\frac{2^k\tau^k}{k!}\Delta.
\]

By setting $\tau=1$ and $k=2\e c\cdot \log n/\log\log n$, for any arbitrary constant $c$,  we get that $\Ex{I(1,k)}=n^{-c}$. Thus, using inequality (\ref{ineq:coup}), we get
\[
\Pr{\mathcal{E}_1}=\Pr{\mathcal{E}_2}\le \Pr{\mathcal{E}_3}\le n^{-c},
\]
which completes the proof. 
\end{proof}

\begin{proof}[Proof of Theorem \ref{thm:rho}]
	By the given construction, consider dynamic graph $\G(n, \rho)$, where for every $t\in \NN$,  $G^{(t)}=H_{k, \lceil 1/\rho\rceil}(A_t, B_t)$, and we set  $k=\Theta(\log n/\log\log n)$.
	At time $t=0$, no node in $B_0$ is informed so by applying  Lemma \ref{lem:string1}, with high probability, at time $t=1$, no node in $S_k$ gets informed, as well. Thus, with high probability, each node in $B_0\setminus (\cup_{i=1}^kS_i)$ is stayed uniformed because by the construction they get informed via  nodes of $S_k$. This implies that, with high probability, 
	\[
	|B_0\cap\I_1|=|B_0\cap (\cup_{i=1}^kS_i)|\le k\lceil 1/\rho\rceil. 
	\]
	Thus,
	\begin{align}\label{ineq:low1}
	B_1=|B_0\setminus \I_1|\ge |B_0|-k\cdot\lceil 1/\rho\rceil= \frac{3n}{4}-k\cdot\lceil 1/\rho\rceil.
	\end{align}
	At time $t=1$,  by definition of $B_1$,  there is no any informed node in $B_1$ and hence we can apply Lemma \ref{lem:string1} and by Inequality (\ref{ineq:low1}), with high probability, we get that
	\[
	B_2=|B_1\setminus \I_2|\ge |B_1|-k\cdot\lceil 1/\rho\rceil\ge  \frac{3n}{4}-2k\cdot\lceil 1/\rho\rceil.
		\]
Inductively, using the  union bound argument, one can easily see that, for every $1\le t\le \frac{n}{4k\cdot\lceil 1/\rho\rceil}$,
\[
|B_t|\ge |B_{t-1}|-k\cdot\lceil 1/\rho\rceil\ge \frac{3n}{4}-t\cdot k \cdot\lceil 1/\rho\rceil\ge \frac{n}{2}.
\]
 Therefore, with high probability, 
 for every $1\le t\le \frac{n}{4k\cdot\lceil 1/\rho\rceil}$
 no node in $B_t$ is informed and $|B_t|\ge n/2$. We now conclude that 
 with high probability,  the  algorithm requires at least \begin{align}\label{ineq:lower}
 	\Omega\left(\frac{n}{4k\cdot\lceil 1/\rho\rceil}\right)=\Omega(n\rho /k)=\Omega\left(\frac{n\rho \cdot \log\log n}{\log n}\right)
 	\end{align}
 	 time to deliver the rumor through $\G(n,\rho)$.
 	 By Observation \ref{obs:hn}, we have  $\Phi(H_{k,\lceil 1/\rho\rceil}(A, B))=\Theta\left(\frac{1}{k+n\rho^2}\right)$ and 
   $\rho(H_{k,\lceil 1/\rho\rceil}(A,B))=\Theta(\rho)$. 
 	 Then, applying Theorem \ref{thm:acc} implies that the rumor spreads in at most 
 	 \[
 	 \Oh\left(\frac{\log n}{\rho\phi}\right)=\Oh\left((\rho n+k/\rho)\log n\right).
 	 \]
 	 Comparing lower bound $\Omega(n\rho/k)$ (i.e., Inequality (\ref{ineq:lower})) and the above upper bound completes the proof.
 	\end{proof}

  \section{The Spread Time in Terms of Absolute Diligence  }\label{sec:abs} In this  section we show Theorem \ref{thm:n2}. 
  For a given nonempty graph $G=(V,E)$, the absolute $\Or(G)$-diligent of $G$ is defined as 
  \[
  \Or(G)=\min_{\{u, v\}\in E}\max\{1/d_u, 1/d_v\}.
  \]
  we define $\Or(G)=0$, if $G$ is an empty graph.
  Let us consider dynamic network
   $\G=\{G^{(t)}\}_{t=0}^\infty$. For every $G^{(t)}\in \G$, we use $\Or(t)$ to denote the  absolute diligence of $G^{(t)}$, $\Or(G^{(t)})$.  
  \begin{proof}[Proof of Theorem \ref{thm:n2}]
  Let $\tau$ denote the first time when all nodes of $\G$ get informed. Similar to the proof of Lemma \ref{lem:key}, for every $\gamma\in[0, \tau)$, we define
  \[
  \lambda(\gamma)=\sum_{\{u,v\}\in E(I_\gamma, U_\gamma)}\left\{\frac{1}{d_u(\gamma)}+\frac{1}{d_v(\gamma)}\right\}\ge \Or(\gamma) \cdot \lceil\Phi(G^{(\gamma)})\rceil,
  \]
  where $\lceil\Phi(G^{(\gamma)})\rceil=1$ if $G^{(\gamma)}=G^{(\lfloor \gamma\rfloor)}$ is connected and  $\lceil\Phi(G^{(\gamma)})\rceil=0$, otherwise. 
  Define  $\{N(\gamma):\lambda(\gamma), \gamma\in[0,\tau)\}$ which is  a non-homogeneous Poisson process counting the number of informed nodes until time $\tau$. By Theorem \ref{thm:npp}, $N(\tau)$ is a Poisson process with rate 
  \begin{align*}
  	\Lambda(\tau)=\int_{0}^{\tau}\lambda(\gamma)d\gamma\ge \sum_{p=0}^{\lfloor \tau \rfloor}\Or(p)  \cdot \lceil\Phi(G^{(\gamma)})\rceil.
  	\end{align*} 
By setting $\tau=T_{abs}(\G,\rho)$, we get $\Lambda(\tau)\ge 2n$. Let $X$ denote a Poisson distribution with rate $2n$. Then, by Lemma \ref{ineq:pois},
\[
\Pr{N(\tau)<n}\le \Pr{X<n}=\e^{n(1/\e-1/2)}.
\]
Therefore, with high probability, $\tau=T_{abs}(\G,\rho)$ is an upper bound for the spread time of the algorithm.
\end{proof}
  \subsection{Absolutely $\rho$-Diligent Dynamic Networks with Spread time $\Theta(n/\rho)$}
  
  In this subsection, for every given $\rho$, we construct a dynamic network $\overline{\G}(n,\rho)=\{G^{(t)}\}_{t=0}^\infty$ where each $G^{(t)}\in \G$ is absolutely $\Theta(\rho)$-diligent and the spread time matches the upper bound given in Theorem \ref{thm:n2} up to a constant factor.
  
\subsubsection*{\underline{Absolutely $\rho$-Diligent  Dynamic Network $\overline{\G}(n,\rho)$}}  

For an arbitrary set of nodes of size $m$, say $A$, and integers $1\le d_1,d_2\le m-1$. 
 Define $G(A, d_1)$ to be a connected $d_1$-regular graph with vertex set $A$.  Also, define    $G(A,d_1,d_2)$ to ba  an $m$-node connected simple graph where each node has degree $d_1$ but one node has degree $d_2$. Note if we choose $d_1$ and $d_2$ to be  even numbers such  graphs exists.
Choose  $\Delta\in \{\lceil1/\rho\rceil, \lceil1/\rho\rceil+1\}$ to be an even number.
Let $V$ denote the vertex set of the dynamic network $\overline{\G}(n, \rho)$ with $|V|=n$. Also, let $A_0\subset V$ and $B_0\subset V$ be the two disjoint subsets of $V$ with
$|A_0|=\lfloor n/2\rfloor$ and $|B_0|=\lceil n/2\rceil$. 
Define $G^{(0)}\in \overline{\G}(n, \rho)$   consisting  of  $G(A_0,4,\Delta)$ and $G(B_0,\Delta)$ where  the node with degree $\Delta$ in
$G(A_0,4,\Delta)$ is connected to an arbitrary node of $G(B_0,\Delta)$.
At time $t=0$, assume that the rumor is injected to a node  of $G(A_0,4,\Delta)$.
The network evolution proceeds in time steps $t\in \NN$.
For every time step $t\in \NN$, set
$B_{t+1}=B_{t}\setminus \I_t$ and $A_{t+1}=V\setminus B_{t+1}$, where $\I_t$ is the set of informed node until time step $t$. If $n/6 \le |B_{t+1}|<|B_{t}|$, then define
 $G^{(t+1)}\in \overline{\G}(n,\rho)$ consisting of  $G(A_{t+1},4,\Delta)$ and $G(B_{t+1},\Delta)$ where  the node with degree $\Delta$ in
 $G(A_{t+1},4,\Delta)$ is connected to an arbitrary node of $G(B_{t+1},\Delta)$.
If $|B_{t+1}|< n/6$, or $|B_{t+1}|=|B_t\setminus \I_t|$, then set  $G^{(t+1)}=G^{(t)}$.
By considering the single edge connecting $G(A_t, 4, \Delta)$ to $G(B_t, \Delta)$, one may observe   that 
	for every $G^{(t)}\in \overline{\G}(n,\rho)$, $\Or(G^{(t)})=1/(\Delta+1)$. Moreover, $\Phi(G^t)=\Oh(1/n)$.

Before showing Theorem \ref{thm:n2}, we need to present some useful lemmas.
The following lemma was proved in \cite{PS13}.
\begin{lem}[Lemmas 9 and 10. \cite{PS13}]\label{kost1}
	For some arbitrary integer number $m>0$, let $f(m, j)$ and $g(m, j)$ be  deterministic sequences such that for $1\le j<m$
	\[
	f(m,j)\le \Ex{t_j|S_j}^{-1}\le g(m,j),
	\]
	where $t_j$ is an exponential random variable  and $S_j$ is some arbitrary random variable. Moreover, let $\{t_j^+\}_{j=1}^{m-1}$ and $\{t_j^-\}_{j=1}^{m-1}$ be a sequences of independent random variables, where $t^+_j$ is exponentially distributed with parameter 
	$f(m,j)$ and  $t^-_j$ is exponentially distributed with parameter $g(m, j)$.  Also let
	$T=\sum_{j=1}^{m-1}t_j$,
	$T^+=\sum_{j=1}^{m-1}t^+_j$ and 
	$T^-=\sum_{j=1}^{m-1}t^-_j$. Then, we have
	\begin{align*}
		\text{for } 0<\lambda< \min_{j\in[m-1]} f(m,j), ~~	\Ex{\e^{\lambda T}}&\le\exp\{\lambda \Ex{T^+}+O(1)\}.\\
		\text{for } \lambda<0, ~~	\Ex{\e^{\lambda T}}&\le\exp\{\lambda \Ex{T^-}+O(1)\}.
	\end{align*} 
\end{lem}	

\begin{lem}\label{removal}
	Consider  $\Delta$-regular graph $G(A,\Delta)$, with vertex set $A$ (which was already defined). Assume, at the beginning,  a node  of $G(A,\Delta)$ is aware of a rumor. Then,   the  algorithm starts propagating the rumor. Fix some arbitrary time $\tau\in[0,1]$, and let  $I_\tau$ counts the number of informed nodes until $\tau\in (0,1]$ in $G(A, \Delta)$. Then, we have $\Ex{I_\tau}=\Theta(1)$ and $\Var{I_\tau}=\Theta(1)$. 
	 
\end{lem}
The proof of the lemma will be given in Subsection  \ref{sec:removal}.
\begin{proof}[Proof of Theorem \ref{thm:n2}]
 Define  $T_0\in \NN$ to be the largest  time step for which $|B_{T_0}\setminus\I_{T_0}|\ge n/6$. 
	By the network construction, for every $t\in \{0,1,\ldots,T_0\}$, at the beginning of each time step,   a non-informed node $v_t\in B_t$ in $G(B_t,\Delta)$ is connected  to  a $\Delta$-degree node $u$ in $G(A_t,4,\Delta)$, where $\Delta\in\{ \lceil1/\rho\rceil,\lceil1/\rho\rceil+1 \}$. Since we are aiming to obtain a lower bound, we may assume that $u$ is always informed and we refer to $v_t$ as the  boundary node.  Define $\tau'_{i}$ to be the time when the $i$-th non-informed boundary node  gets informed, which is meaningful by the network definition. Also, let $\tau'_0=0$. For every $i\ge 1$, define random variable $\tau_i=\min\{\tau'_i, T_0\}$. For each $t\in \{1,2\ldots,T_0\}$, the boundary node $v_t$ pulls the rumor from $u$ with exponential time clock of rate $1/(\Delta+1)$ and $u$ pushes the rumor with rate $1/(\Delta+1)$. Therefore, for each $t\in\{1,2\ldots,T_0\}$, the waiting time for $v_t$ to become informed is $(\Delta+1)/2$. Then, for every $i\ge 1$, 
	$\Ex{\tau_{i}-\tau_{i-1}} \in\{
	(\Delta+1)/2,0\}$.
		Define random variable $X_i\ge 1$ to be  the  number of  nodes in $G(B_t,\Delta)$ that get informed   during time interval	$[\tau_{i}, \lceil \tau_{i}\rceil]$.  Applying  Lemma \ref{removal} we infer that  
	\[
	\Ex{X_i}=\mu_i=\Theta(1), \Var{X_i}=\sigma^2_i=\Theta(1).
	\]
 Set $n_0={n}/{(10+10\mu)}$, where $\mu=\max_{i=1}^{n_0} \mu_i$, and define $X=\sum_{i=1}^{n_0}X_i$. Then we have $n_0\le \Ex{X}\le n_0\mu\le n/10$.  $X_i$'s, $1\le i\le n_0$, are independent with bounded variance and  hence using  Chebychev's inequality yields that 
 \[
 \Pr{X\ge 2\Ex{X}}\le \Pr{|X-\Ex{X}|\ge\Ex{X}}\leq \frac{\Var{X}}{\Ex{X}^2}=\Oh(1/n)
 \]
 Therefore, with probability $1-\Oh(1/n)$, $X<2\Ex{X}\le n/5$.
 So, with probability $1-\Oh(1/n)$, we conclude that \[
 |B_{\lfloor\tau_k\rfloor+1}|\ge n/2-X\ge 3n/10> n/6.\]
 Define $\E(n_0)$ to be the event that $|B_{\lfloor\tau_{n_0}\rfloor+1}|> n/6$. Conditional on $\E(n_0)$, for every $1\le i \le n_0$,  we have  $\tau_i < T_0$ and hence,
   \[\Ex{\tau_{i}-\tau_{{i-1}}|\E(n_0)}^{-1}\le 2/\Delta.\]
For each $i=1,\ldots n_0$, let $Z_i$ to  denote an exponential random variable with parameter $2/\Delta$. 
	Consider the random variable $Z=\sum_{i=1}^{n_0} Z_i$. Then we have that $\Ex{Z}=\Ex{\sum_{i=1}^{n_0} Z_i}={n_0}\Delta/2$. We  have defined $\tau_0=0$, so  we have $\sum_{i=1}^k\{\tau_{i}-\tau_{i-1}\}=\tau_k$. Applying 
Lemma \ref{kost1}
 implies that,  for every $\lambda<0$, $\Ex{\e^{\lambda\tau_k}}\le \exp\{\lambda\Ex{Z}+\Oh(1)\}$. Thus,  
	\begin{align*}
		&\Pr{\tau_{n_0}< n_0\Delta /4}=\Pr{\e^{\lambda\tau_{n_0}}>\e^{\lambda n_0\Delta/4}}
		\le \frac{\Ex{\e^{\lambda\tau_{n_0}}}}{\e^{\lambda n_0\Delta/4}}\\&\le
		\exp\{\lambda\Ex{Z}+\Oh(1)-\lambda n_0\Delta/4\}=
			\exp\{\lambda n_0\Delta/4+\Oh(1)\}, 
	\end{align*}
	where the inequality comes from  Markov's inequality. 
	Therefore, conditional on $\E(n_0)$,  with high probability  $\tau_{n_0}=\Omega(n_0\Delta)=\Omega(n_0/\rho)$, which follows from definition of $\rho$.
	Since $\Pr{\E(n_0)}\ge 1-O(1/n)$, then with probability $1-\Oh(1/n)$,  $\tau_{n_0}=\Omega(n/\rho)$. On the other hand, $\tau_{n_0}$  is a lower bound for the spread time and the proof is completed.
\end{proof}

\subsection{Proof of Lemma \ref{removal}}\label{sec:removal}
\begin{proof}
	In order to simplify the analysis, we consider  asynchronous $2$-push algorithm, where each node is associated with an exponential time clock of rate $2$, and if  the clock of an informed node rings, the node sends rumor to a random neighbor. This process has the same performance on regular graphs as the original algorithm. Because the both process pick an edge with rate $2/\Delta$.
 For some arbitrary $j\ge1$, assume that  $j$ nodes are informed in the graph and hence the $j$ informed node has at most $j\Delta$ uninformed neighbors and each neighbor is picked with Poisson rate $2/\Delta$.  Let $t_j$, $1\le j< m=|A|$, denotes the time required to inform the $(j+1)$-th node. Thus, 
 	\begin{align*}\label{upper:example}
 	&\left(\Ex{t_j | \text{$j$ informed nodes}}\right)^{-1}\leq j\Delta(2/\Delta)=2j.
 \end{align*}
Let $t^-_j$ denote an exponentially distributed random variable with rate $2j$. For each $k=1,\ldots, m$, let  $T^-_k=\sum_{j=1}^k{t^-_j}$ and  by linearity of expectation,  we get 
\[
\Ex{T^-_k}=\sum_{j=1}^{k}\Ex{t^-_j}=\sum_{j=1}^{k}\frac{1}{2j}={H_k}/{2},	
\] 
where $H_k$ is the $k$-th harmonic number.  By the	Markov inequity we have 
\[
\Pr{I_\tau\ge k}\le \Pr{T_k\le \tau}=\Pr{\e^{\lambda{T_k}}\ge \e^{\lambda\tau}}
\le \Ex{\e^{\lambda T_k}}\e^{-\lambda\tau}
\]

Moreover, applying  Lemma \ref{kost1} with $\lambda=-6$ we gets that 
	\begin{align*}
		\Pr{I_\tau\ge k}&\le \Ex{\e^{\lambda T_k}}\e^{-\lambda\tau}\le  \e^{\lambda\Ex{T_k^-}+\Oh(1)-\lambda\tau}\\&\le \e^{\lambda\Ex{T_k^-}+\Oh(1)-\lambda}=\e^{-3H_k+\Oh(1)},
	\end{align*}
which follows from the fact that $\tau\in [0,1]$. Since for every positive integer $k$, $H_k=\log k+\Oh(1)$, we get that \[\sum_{k=1}^{\infty}\e^{-3H_k+\Oh(1)}\le \sum_{k=1}^{\infty} \Oh(1/k^3)=\Oh(1).\] 
This follows that $
	\Ex{I_\tau}=\sum_{k=1}^{\infty}\Pr{I_\tau\ge k}\le\sum_{k=1}^{\infty}\e^{-3H_k+\Oh(1)}=\Oh(1)$.
 Moreover,
	\begin{align*}
		\Ex{I_\tau^2}&=\sum_{k=1}^\infty\Pr{I_\tau^2\ge k}
		=\sum_{k=1}^\infty\Pr{I_\tau\ge \sqrt{k}}\\&=\e^{-3H_{\sqrt{k}}+\Oh(1)}=\sum_{k=1}^{\infty} \Oh(k^{-3/2})=\Oh(1).
	\end{align*}
	This implies that $\Var{I_\tau}\le \Ex{I_\tau^2}=\Oh(1)$.	\end{proof}

 \section{ Asynchronous versus Synchronous in Dynamic  Networks}\label{asyvs}
 In this section we compare the spread time of the synchronous and asynchronous rumor spreading in dynamic networks $\G_1$ and $\G_2$, presented in Figure (\ref{Fig1}), and show Theorem \ref{thm:asynch-star}. Recall that the synchronous algorithm proceeds in successive rounds, in each round, every node contacts a random neighbor and they exchange the rumor if at least one of them knows it. All nodes follow the synchronized rounds.  


 \begin{proof}[ Proof  of Theorem \ref{thm:asynch-star}  Part (i)]
The graph $G^{(0)}\in\G_1=\{G^{(0)}, G^{(1)}\ldots\}$ is an $n$-node  clique with  the pendent edge $\{1, n+1\}$, where node $n+1$ is only connected to node $1$.
For every $t\ge 1$, $G^{(t)}=G^{(1)}$ and $G^{(1)}$ consists of two  equally-sized cliques connected by an edge, which we refer to it as the bridge. Moreover,  nodes $1$ and $n+1$ are contained in those cliques, called   the left and  right clique  (see Fig \ref{Fig1} (a)). Initially, a rumor is injected  to node $n+1$ in $G^{(0)}$. In the asynchronous algorithm, each node has its own clock of rate $1$, so  with constant probability nodes $1$ and $n+1$ do not contact each other in time interval $[0,1)$. Therefore, with constant probability, at  $t=1$,   nodes contained in the left clique do not know the rumor and they have to wait until the rumor is delivered via the bridge. Also, the bridge is picked according to an exponential time clock of rate $2/n+2/n$. So, it takes $\Omega(n)$ time for a node of the left clique to become informed. Hence $T_a(\G_1)=\Omega(n)$.

In synchronous  algorithm, in the first round, node $n+1$ pushes the rumor to node $1$, with probability $1$. Therefore,
for every $t\ge 1$, nodes $1$ and $n+1$ know the rumor and 
  nodes of  the both cliques of $G^{(t)}$ simultaneously get informed in $\Oh(\log n)$, which means $T_s(\G_1)=\Theta(\log n)$.  
Note that this is well known that the synchronous  algorithm spread the rumor in a clique  of size $n$ after $\Theta(\log n)$ rounds (time), w.h.p. \cite{KSSV00,Boyd2006}.
\end{proof}
\begin{proof}[ Proof  of  Theorem \ref{thm:asynch-star}  Part (ii)]
 $\G_2=\{G^{(0)}, G^{(1)},\ldots\}$  is a sequence of stars and evolves as follows.
$G^{(0)}$ is a star with $n+1$ nodes and the rumor is injected to an arbitrary leaf node. In each time step $t\in \NN$, the center which might be got informed is replaced by an uninformed leaf node. If there is no any uninformed leaf node, then  the center is replaced with a random  leaf node.
Consider  the synchronous  algorithm  whose steps are synchronized with the dynamics of the network. Therefore, in a round, if the center becomes informed by either push or pull call, the other uniformed leaves cannot pull  the rumor from the center, at the same round, because any action is allowed to be taken at the beginning of each round. So we have $
T_s(\G_2)=n$.
The dynamic star is an expander graph and   $1$-diligent. Therefore, applying Theorem \ref{thm:acc} gives  upper bound $\Oh(\log n)$ for the spread time. Moreover, by the exponential time distribution,  the first informed node's clock ticks after $\Omega(\log n)$ time, with high probability. Therefore, $T_a(\G_2)=\Theta(\log n)$.

\end{proof} 
\subsection{ Proof of   Theorem \ref{thm:asynch-star} Part (iii) }
In order to show the part (iii), we analyze the algorithm in $\G_2$ in two consecutive phases. The proof of Lemmas \ref{lem:p1} and \ref{lem:p2} can be found in Appendix \ref{app:asy}.
\subsubsection*{First Phase}
This phase starts with a single informed node and completes when $\Omega(n)$ nodes get informed.
We use $t_f$ to denote the time when this phase competes. 

\begin{lem}\label{lem:p1} With probability at least $1-\e^{-k/2-o(1)}$, $t_f\le k$. 
\end{lem}
The proof is based on the fact that, at some time interval $[t_f, t_f+1]$, an informed leaf pushes the rumor to the central node in time interval $[t_f, t_f+c)$ for some constant $c\in(0,1)$. Then in time interval $[t_f+c, t_f+1)$, $\Omega(n)$ leaves pull the rumor from the center.Moreover $t_f$ has a geometric probability distribution 
\subsubsection*{Second Phase}
This phase starts with $\Omega(n)$ nodes and ends when all $n+1$ node get informed. Let $t_s$ denote the time when the phase completes.

\begin{lem}\label{lem:p2}
	With probability $1-\e^{-k-o(1)}$, $t_f-t_s\le k$.
\end{lem}
Since there are $\Omega(n)$ informed nodes, one of them pushes the rumor to the central node with rate $\Omega(n)$ and then rest of the uninformed nodes pull the rumor from the center.

\begin{proof}[Proof of Theorem \ref{thm:asynch-star} part(iii)]
	Combining the results obtained in Lemmas \ref{lem:p1} and \ref{lem:p2} shows that
	\[
	\Pr{t_s>2k}\le \Pr{t_f>k \text{~or~} t_s-t_f>k }\le  \Pr{t_f>k}+\Pr{t_s-t_f>k}=\e^{-k/2-o(1)}+\e^{-k-o(1)},
	\]
	which completes the proof. 
\end{proof}

 	\bibliography{diss1}
 	\appendix
 
 	\section{Missing proofs of Section \ref{asyvs}}\label{app:asy}
 	Let us first recall  a Chernoff bound.
 	 \begin{theorem}[Chernoff Bounds]\label{lem:cher}
 		Suppose that $X_1, X_2,\ldots, X_n\in \{0, 1\}$ are independent random variables and let $X=\sum_{i=1}^n X_i$. 
 		Then for every $\delta\in (0, 1)$ the following inequalities hold
 		\begin{align*}
 			\Pr{X\ge (1+\delta)\Ex{X}} &\le \exp(-\delta^2 \Ex{X}/2),\\
 			\Pr{X\le (1-\delta)\Ex{X}} &\le \exp(-\delta^2 \Ex{X}/3).
 		\end{align*}
 		In particular,
 		\[
 		\Pr{|X-\Ex{X}|\ge \delta\Ex{X}}\le 2\exp(-\delta^2 \Ex{X}/3).
 		\]
 	\end{theorem}
 	For a proof see \cite{DP09}.

 	\begin{proof}[Proof of Lemma \ref{lem:p1}]
 		Let $c\in(0,1)$ be a constant that will be fixed later. Suppose that we are at some time $t=0,1\ldots$ and let $v_t$ be the central node of the star  within time interval $[t, t+1)$. 
 		On the other hand, for every discrete time $t=0,1,\ldots$, suppose that we are at time $t$ and define  $\tau^t_{w\rightarrow v_t}$ to denote the first time when an informed leaf node, say $w$, pushes the rumor to  $v_t$. 
 		For every  $t\in \NN$, $\tau^t_{w\rightarrow v_t}$, $w$ has an exponential distribution with rate $1$. Therefore, for every $t\in \NN$   
 		\begin{align}
 			p_c&=\Pr{  \tau^t_{w\rightarrow v_t}\le t+c}= 1-\Pr{\tau^t_{w\rightarrow v}> t+c}\nonumber\\
 			&= 1-\frac{\int_{t+c}^{\infty}\mathrm{e}^{-s} ds}{\int_{t}^{\infty}\mathrm{e}^{-s} ds}
 			=1-\mathrm{e}^{-c}.
 		\end{align}
 		Let $Z_c$ denote a  geometric distribution with success probability $p_c$, then 
 		one may easily see that  the number of rounds required for $w$ to inform a central node within $[t, t+c)$ time is stochasticaly dominated by $Z_c$ . For each  leaf node $u$,
 		let $X^t_c(u)$ to denote an indicator random variable taking one if $u$'s clock ticks 	within interval $[t+c, t+1)$ and contacts  $v_t$.
 		Thus, we have 
 		\[
 		\Pr{X^t_c(u)=1}=\frac{\int_{t+c}^{t+1}\mathrm{e}^{-s} ds}{\int_{t}^{\infty}\mathrm{e}^{-s} ds}=\mathrm{e}^{-c}-\mathrm{e}^{-1}.
 		\]
 		which follows from the fact that $X^t_c(u)$ is only considered after time $t$.
 		Let $Y^t_c=\sum_{u \text{ is a leaf}}X^t_c(u)$ counts the number of leaf nodes  contacting $v_t$ within interval $[t+c, t+1)$.
 		Applying the linearity of expectation we get
 		\[
 		\Ex{Y_c^t}=\sum_{u \text{ is a leaf}}\Ex{X^t_c(u)}=n(\mathrm{e}^{-c}-\mathrm{e}^{-1}).
 		\]
 		Using a Chernoff bound  (e.g., see Theorem \ref{lem:cher}) yields that, with probability $1-n^{-\omega(1)}$, $Y_c^t=n(e^{-c}-\e^{-1})(1\pm o(1))$. Therefore,  with probability $1-n^{-\omega(1)}$, after the first success of $Z_c$, $\Omega(n)$ nodes get informed which implies that 
 		\[
 		\Pr{t_f> k}\le \Pr{Z_c> k}(1-o(1))\le (1-o(1))(1-p_c)^{k}
 		\]
 		By setting $c=2/3$ we have that $\Pr{t_f>k}\le \e^{-k/2-o(1)}$ completing the proof.
 	\end{proof}
 
 	\begin{proof}[Proof of Lemma \ref{lem:p2}]
 		By the end of the first phase,   there  are  at least $n/\log n$ informed nodes. Supposed that we are at time $t=\lceil t_f\rceil,\ldots$ and hence the probability that an informed node pushes the rumor to central node within interval $[t, t+1/\log n]$ is 
 		\[
 		\frac{\int_{t}^{t+1/\log n}\e^{-s}ds}{\int_{t}^{\infty}\e^{-s}ds}=1-\e^{-1/\log n}.
 		\]
 		Since every informed nod has an independent clock of rate $1$. The probability that none of the informed nodes pushes the rumor to the center before $t+1/\log n$ is at most $
 		\e^{-n/\log^2 n}=n^{-\omega(1)}$.	This implies that, for every time $t>\lceil t_f\rceil$  with probability $1-n^{-\omega(1)}$, the central node becomes informed during interval $[t, t+1/\log n]$.
 		For every  leaf node  $u$ which is not informed until $t+1/\log n$, the probability that  $u$ gets informed in time interval $[t+1/\log n, t+1)$ is
 		\begin{align}
 			&\Pr{X_t(u)=1}=\left(\frac{\int_{t+1/\log n}^{t+1}\e^{-s}ds}{\int_{t}^{\infty}\e^{-s}ds}\right)\left(1-n^{-\omega(1)}\right)\nonumber\\&=(\e^{-1/\log n}-\e^{-1})(1-o(1))=1-\e^{-1}-o(1)=p
 		\end{align}
 		The first multiplier the probability that the $u$'s clock ticks in $[t+1/\log n, t+1]$ and the second one is the probability that the central node gets informed in $[t, t+1/\log n]$. Let $Z_p$ be a geometric random variable with success probability $p$ then $Z_p$ dominates $t_s-t_f$. Hence,
 		\[
 		\Pr{t_s-t_f>k}\le \Pr{Z_p>k}\le (1-p)^k\le  (\e-o(1))^{-k}=\e^{-k-o(1)}.
 		\]
 	\end{proof}

 	\end{document}